\documentclass[submission,copyright,creativecommons]{eptcs}
\providecommand{\event}{TARK 2021} %
\usepackage{underscore}

\newenvironment{proof}{\paragraph{Proof.}}{\hfill$\square$}

\usepackage{hyperref}

\usepackage{url}

\usepackage{tikz}
\usetikzlibrary{positioning}

\usepackage[ruled,vlined,linesnumbered,lined]{algorithm2e}

\usepackage{amsmath}
\usepackage{amsfonts}
\usepackage{amssymb}
\usepackage{latexsym}
\usepackage{alltt}
\usepackage{thm-restate}
\usepackage{graphics}
\usepackage{subcaption}

\usepackage{xspace}
\usepackage{graphicx,color}

\usepackage{enumerate}
\usepackage{color}
\usepackage{sgame}
\usepackage{microtype}

\newtheorem{theorem}{Theorem}
\newtheorem{defined}[theorem]{Definition}

\newtheorem{exa}[theorem]{Example}
\newenvironment{example}{\begin{exa} \rm}{\end{exa}}

\newtheorem{corollary}[theorem]{Corollary}

\newtheorem{exe}{Exercise}

\newtheorem{pro}{Problem}

\newcommand{\bfe}[1]{\begin{bfseries}\emph{#1}\end{bfseries}}

\newcommand{\ES}{\mbox{$\emptyset$}}

\newcommand{\myra}{\mbox{$\:\rightarrow\:$}}

\newcommand{\La}{\mbox{$\:\Leftarrow\:$}}
\newcommand{\Ra}{\mbox{$\:\Rightarrow\:$}}

\newcommand{\sse}{\mbox{$\:\subseteq\:$}}

\newcommand{\fa}{\mbox{$\forall$}}
\newcommand{\te}{\mbox{$\exists$}}

\newcommand{\LL}{\mbox{$\ldots$}}

\newcommand{\NI}{\noindent}
\newcommand{\HB}{\hfill{$\Box$}}

\newcommand{\II}{\vspace{2 mm}}

\newcommand{\szkew}[1]{\relax \setbox0=\hbox{\kern -24pt $\displaystyle#1$\kern 0pt }%
\box0}
{\catcode`\@=11 \global\let\ifjusthvtest@=\iffalse}

\newcounter{oldmycaption}

\newcommand{\turn}{\mathit{turn}}

\newcommand{\leaf}{\mathit{leaf}}
\newcommand{\win}{\mathit{win}}
\newcommand{\spe}{\mathit{SPE}}
\newcommand{\NE}{\mathit{NE}}
\newcommand{\draw}{\mathit{draw}}

\title{Well-Founded Extensive Games with Perfect Information}
\author{Krzysztof R. Apt
	\institute{Centrum Wiskunde \& Informatica\\ Amsterdam, The Netherlands}
	\institute{University of Warsaw\\ Warsaw, Poland}
    \email{k.r.apt@cwi.nl}
	\and
        Sunil Simon
	\institute{Department of CSE,\\ IIT Kanpur, Kanpur, India} 
	\email{simon@cse.iitk.ac.in}
}
\def\titlerunning{Well-Founded Extensive Games with Perfect Information}
\def\authorrunning{K.~R.~Apt \& S.~Simon}

\begin{document}

\maketitle

\begin{abstract}
  We consider extensive games with perfect information with
  well-founded game trees and study the problems of existence and of
  characterization of the sets of subgame perfect equilibria in these
  games. We also provide such characterizations for two classes
    of these games in which subgame perfect equilibria exist: two-player
    zero-sum games with, respectively, two and three outcomes.  
\end{abstract}

\section{Introduction}

Research on strategic games assumes that players have to their
disposal infinitely many strategies. This allows one to view strategic
games with mixed strategies as customary strategic games with
infinitely many strategies. This assumption is also used in a study
of various standard examples, such as Cournot or Bertrand competition,
in which the players have to set the production level or the price of
a product.

In contrast, the exposition of the standard results for the extensive
games with perfect information (from now, just `extensive games') is
usually limited to finite games. This restriction rules out a study of
various natural examples, for example infinite variants of the
ultimatum game or some bargaining games, see, e.g., \cite{Rit02}. In
the first case the game has just two stages but the first player may
have infinitely many actions to choose from, while in the second case
the game has an arbitrary, though finite, number of stages. Such games
are then analyzed separately, without taking into account general
results.

The aim of this paper is to provide a systematic account of extensive
games with perfect information in a setting that only requires that
the underlying game tree is well-founded (i.e., has no infinite
paths). We call such games \emph{well-founded}\footnote{Such games are
  sometimes called games with \emph{finite horizon} (see, e.g.,
  \cite{OR94}). We decided to use instead the qualification
  `well-founded' because `finite horizon' is sometimes used to
  indicate that the game tree is of bounded depth, i.e, has a finite
  rank (a concept introduced in the next section).}.

The standard tool to analyze finite extensive games is the concept of
a subgame perfect equilibrium. Their existence is established by means
of the backward induction algorithm.  In infinite well-founded
extensive games subgame perfect equilibria may fail to exist. Also
one cannot resort to any version of this algorithm since it
will not terminate. In principle this could be taken care of by
defining a joint strategy as an eventual outcome of an infinite
computation.  However, for arbitrary well-founded games one would have
then to proceed by means of a transfinite induction, which raises a
legitimate question whether such a process can be called a
computation.

Therefore, instead of trying to define such generalized computations
we dispense with the backward induction altogether and simply proceed
by transfinite induction. This results in mathematical proofs of
existence that are not supported by any algorithm, but still, as
illustrated by examples, the obtained results can be used to compute
the sets of subgame perfect equilibria in specific well-founded games
and to deduce their existence in special cases. Informally,
transfinite induction analyzes the game tree `top down', while the
backward induction proceeds 'bottom up' and consequently cannot be
naturally applied to infinite game trees.

Most results, though not all, are natural generalizations of the
corresponding results for the case of finite extensive games. Some of
these results fail to hold for infinite games and some of the
traditional proofs for finite games, notably the ones involving the
backward induction, have to be suitably modified.  In what follows we
focus both on arbitrary well-founded games and on two-player
  zero-sum games with, respectively, two and three outcomes.  

In the literature we found only one paper in which well-founded games
appear, namely \cite{EO12}. The authors provide using higher-order
computability theory a formula that defines the set of subgame perfect
equilibria under an assumption that implies their existence, and apply
it to determine a subgame perfect equilibrium in an infinite three
stage game.  In several books various examples of infinite extensive
games are studied and various extensions of finite extensive games,
for example games with chance moves, see, e.g., \cite{Rit02},
simultaneous moves, see, e.g., \cite{OR94}, or repeated extensive
games, see \cite{MS06}, are introduced (not to mention games with
imperfect information).  Also subgame perfect equilibria in games
allowing infinite plays have been studied, see, e.g., \cite{FL83} and
a more recent \cite{Kam19}.  In \cite{aFR16} a maximally general
definition of an extensive form game is proposed that among others
covers repeated games, differential games, and stochastic games.  In
the proposed framework even immediate predecessors of an action may
not exist (like in continuous time interactive decisions examples).
In our opinion the class of games considered here merits attention as
a first natural generalization to study.

In the next section we introduce the relevant concepts and provide
natural examples of well-founded extensive games. In Section
\ref{sec:spe} we establish existence of subgame perfect equilibria for
some natural classes of well-founded games and show how to apply a
characterization result to compute the set of subgame perfect
equilibria for specific example games.  Then, in Section
\ref{sec:winorlose} we consider two classes of two-player well-founded
games: win or lose games and chess-like games.  As a stepping stone
towards characterizations of the sets of subgame perfect equilibria in
these games we show that the well-known result attributed to Zermelo
\cite{Zer13} (see also \cite{SW01}) about existence of winning
strategies continues to hold for well-founded games.

\section{Preliminaries on extensive games}

A \bfe{tree} is an acyclic directed connected graph, written as
$(V,E)$, where $V$ is a non-empty set of nodes and $E$ is a possibly
empty set of edges. In drawings the edges will be directed downwards.

An \bfe{extensive game with perfect information} (in short, just
an \bfe{extensive game})  for $n \geq 1$ players consists of:

\begin{itemize}

\item a set of players $\{1, \LL, n\}$,

\item a \bfe{game tree}, which is a tree $T := (V,E)$ with a \bfe{turn
    function} $turn: V \setminus Z \to \{1, \LL, n\}$, where $Z$ is
  the set of leaves of $T$,

\item the \bfe{payoff functions}
$p_i: Z \myra \mathbb{R}$, for each player $i$.
\end{itemize}
We denote it by $(T, turn, p_1, \LL, p_n)$.

The function $turn$ determines at each non-leaf node which player
should move.  The edges of $T$ represent possible \emph{moves} in the
considered game, while for a node $v \in V \setminus Z$ the set of its
children $C(v) := \{w \mid (v,w) \in E\}$ represents possible
\bfe{actions} of player $turn(v)$ at $v$. For a node $u$ in $T$ let
$T^u$ denote the subtree of $T$ rooted at $u$.

We say that an extensive game is \bfe{finite}, \bfe{finite depth},
\bfe{infinite}, or \bfe{well-founded} if, respectively, its game tree
is finite, finite depth, infinite, or well-founded. Recall that a tree
is called \bfe{well-founded} if it has no infinite paths (see, e.g.,
\cite[page 224]{TvD88}).

Further, following \cite{Bat97}, we say that an extensive game is
\bfe{without relevant ties} if for all non-leaf nodes $u$ in $T$
the function $p_i$, where $\turn(u)=i$, is injective on the
  leaves of $T^u$. This is more general than saying that a game is
  \bfe{generic}, which means that each $p_i$ is an injective
  function.
  
We shall often rely on the
concept of a \emph{rank} of a well-founded tree $T$. Recall that it is
defined inductively as follows, where $v$ is the root of $T$:
  \[
\mathit{rank}(T):=
  \begin{cases}
    0 &\text{ if $T$ has one node}\\
    \mathit{sup} \{\mathit{rank}(T^u) + 1 \mid u \in C(v) \} &\text{
      otherwise,}
  \end{cases}
\]
where $sup(X)$ denotes the least ordinal larger than all ordinals in
the set $X$.
Transfinite induction will be needed only to deal with
games on the trees with rank $> \omega$.

In the figures below we identify the actions with the labels we put
on the edges and thus identify each action with the corresponding
move. For convenience we do not assume the labels to be unique, but it
will not lead to confusion.  Further, we annotate the non-leaf nodes
with the identity of the player whose turn it is to move and the name
of the node. Finally, we annotate each leaf node with the
corresponding sequence of the values of the $p_i$ functions.

\begin{example} \label{exa:ultimatum}

  The following two-player game is called the \bfe{Ultimatum game}.
  Player 1 moves first and claims a real number $x \in [0, 100]$, to
  be interpreted as a fraction of some good to be shared, leaving the
  fraction $100-x$ for the other player. Player 2 either
  \emph{accepts} this decision, the outcome is then $(x, 100-x)$, or
  \emph{rejects} it, the outcome is then $(0,0)$. 
  The game tree is depicted in Figure \ref{fig:ultimatum}, where the
  action of player 1 is a number from the set $[0, 100]$, and the
  actions of player 2 are denoted by $A$ and $R$.  The resulting game
  is infinite but is of finite depth. The rank of the game tree is 2.

\begin{figure}[ht]
  \centering
  \tikzstyle{level 1}=[level distance=1.5cm, sibling distance=5cm]
  \tikzstyle{level 2}=[level distance=1.5cm, sibling distance=2.5cm]
  \tikzstyle{level 3}=[level distance=1.5cm, sibling distance=2cm]
\begin{tikzpicture}
 \node (r){1, $u$}
 child{
   node (a){2, $0$}
   child{
     node (d){$(0, 100)$}
     edge from parent
     node[left]{\scriptsize $A$}
   }
   child{
     node(e){$(0,0)$}
     edge from parent
     node[right]{\scriptsize $R$}
     edge from parent
   }
   edge from parent
   node[left]{\scriptsize $0$}
   }
 child{
   node (b){2, $x$}
   child{
     node (f){$(x, 100-x)$}
     edge from parent
     node[left]{\scriptsize $A$}
   }
   child{
     node (g){$(0,0)$}
     edge from parent
     node[right]{\scriptsize $R$}
     edge from parent
   }
   edge from parent
   node[left]{\scriptsize $x$}
 }
  child{
   node (c){2, 100}
   child{
     node (h){$(100, 0)$}
     edge from parent
     node[left]{\scriptsize $A$}
   }
   child{
     node (i){$(0,0)$}
     edge from parent
     node[right]{\scriptsize $R$}
     edge from parent
   }
  edge from parent
  node[right]{\scriptsize $100$}};

  \path (a) -- (b) node [midway] {$\cdots$};
  \path (b) -- (c) node [midway] {$\cdots$};
\end{tikzpicture}
    \caption{The Ultimatum game}
    \label{fig:ultimatum}
  \end{figure}
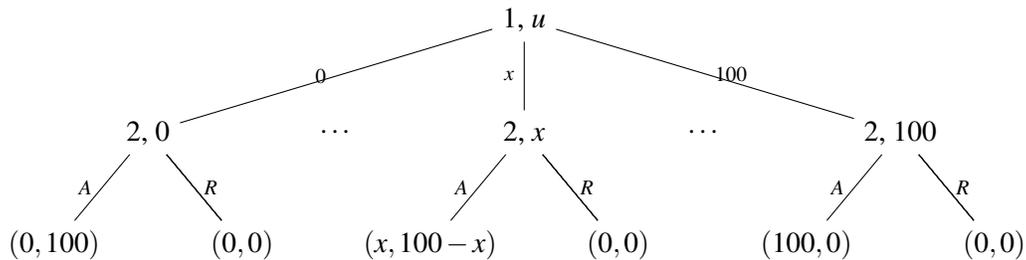
\HB  
\end{example}

\begin{example} \label{exa:well-founded}

  Consider the following \bfe{Bargaining game} (without depreciation).
  Player 1 moves first by selecting a natural number $k \geq 2$. Such
  a choice is to be interpreted that he claims the fraction $1 -
  \frac{1}{k}$ of some good to be shared, leaving the fraction
  $\frac{1}{k}$ for the other player. As long as player 1 selects $k >
  2$, player 2 asks for a better offer or rejects it.  In the first
  case player 1 selects $k-1$. The game continues until player 1
  selects 2, i.e., claims 50 \% of the good.  At that moment player 2
  either accepts this offer, the outcome is then $(50,50)$, or rejects
  it. All rejections result in the outcome $(0,0)$.  The game tree of
  this game, depicted in Figure~\ref{fig:bargain}, has arbitrary long,
  though finite, branches, so this game is infinite but it is
  well-founded.  The rank of the game tree is $\omega$.
  \HB
\end{example}

\begin{figure}[ht]
  \centering
  \tikzstyle{level 1}=[level distance=1.5cm, sibling distance=5cm]
  \tikzstyle{level 2}=[level distance=1.5cm, sibling distance=2.5cm]
  \tikzstyle{level 3}=[level distance=1.5cm, sibling distance=2cm]
  \begin{tikzpicture}
    [
 scale=1.5,font=\footnotesize,
 level 1/.style={level distance=12mm,sibling distance=25mm},
 level 2/.style={level distance=10mm,sibling distance=15mm},
 level 3/.style={level distance=10mm,sibling distance=10mm},
 level 4/.style={level distance=10mm,sibling distance=10mm}
]
 \node (r){1}
 child{
   node (a){$2, 2$}
   child{
     node (d){$(50, 50)$}
     edge from parent
     node[left]{$A$}
   }
   child{
     node(e){$(0,0)$}
     edge from parent
     node[right]{ $R$}
   }
   edge from parent
   node[left]{$2$}
   }
 child{
   node (b){$2, 3$}
   child{
     node (f){$1,B$}
     child{
       node(fo){$2,3 \cdot 2$}
       child{
     node (foo){$(50, 50)$}
     edge from parent
     node[left]{ $A$}
   }
   child{
     node(fot){$(0,0)$}
     edge from parent
     node[right]{ $R$}
   }
       edge from parent
       node[left]{ $2$}
       }
     edge from parent
     node[left]{ $B$}
   }
   child{
     node (g){$(0,0)$}
     edge from parent
     node[right]{ $R$}
     edge from parent
   }
   edge from parent
   node[left]{ $3$}
 }
  child{
    node (c){$2, k$}
  edge from parent
  node[right]{$k$}};
  \node [below =of c] {$\cdots$};
    \node [right =of c] {$\cdots$};
  \path (b) -- (c) node [midway] {$\cdots$};
\end{tikzpicture}
    \caption{The Bargaining  game}
   \label{fig:bargain}
\end{figure}
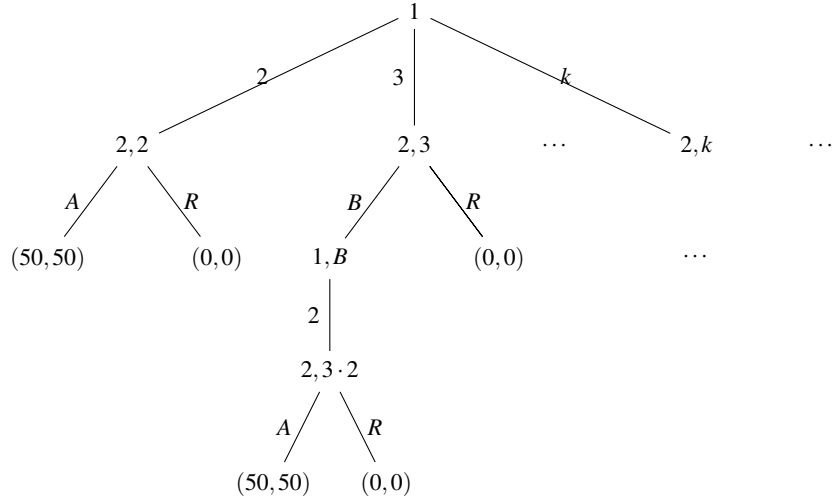

Below, given a two-player extensive game we denote the opponent of
player $i$ by $-i$ instead of $3-i$.

\begin{example} \label{exa:alpha}
  We now construct a sequence of games $G(i, \alpha)$, where
  $i \in \{1, 2\}$ and $\alpha$ is an ordinal $> 1$, by induction as
  follows:

   \begin{itemize}
   \item $G(1,2)$ is the the ultimatum game from Example
     \ref{exa:ultimatum} and $G(2,2)$ is its version with the roles
     of players 1 and 2 reversed;
     
   \item $G(i, \alpha)$, where $\alpha > 2$, is obtained as follows:

     \begin{itemize}

     \item its game tree is constructed by selecting a root $v$, 
taking the game trees of the games $G(-i, \beta)$, 
where $1 < \beta < \alpha$ and selecting their roots as the children of $v$,

     \item setting $turn(v) = i$.

     \end{itemize}

   \end{itemize}

   So for example the root of the game tree of $G(i,3)$ has one child,
   namely the root of the game tree of $G(-i,2)$, the root of the game
   tree of $G(i,4)$ has two children, namely the roots of the game
   trees of $G(-i,3)$ and $G(-i,2)$, etc.  Note that the rank of the
   game tree of $G(i, \alpha)$ is $\alpha$.  \HB
\end{example}

The class of endogenous games studied in \cite{JW05} form another
example of well-founded extensive games. These games are played in two
stages. In the first stage the players are involved in pre-play
negotiations that essentially fix the payoff functions and in the
second stage they choose their strategies. The resulting game is
infinite due to the pre-play negotiations, while the rank of the game
tree is 2.

Note that by K\"onig's lemma \cite{Kon27} every finitely branching
well-founded extensive games is finite.  Consequently, interesting
well-founded extensive games necessarily have infinite branching.

For an extensive game $G:= (T, turn, p_1, \LL, p_n)$ let 
$V_i := \{v \in V \setminus Z \mid turn(v) = i\}$. So $V_i$ is the
set of nodes at which player $i$ moves.  A \bfe{strategy} for player
$i$ is a function $s_i: V_i \to V$, such that $(v, s_i(v)) \in E$ for
all $v \in V_i$.  We denote the set of strategies of player $i$ by
$S_i$.

Let $S = S_1 \times \cdots \times S_n$.  We call each element $s \in
S$ a \bfe{joint strategy}, denote the $i$th element of $s$ by $s_i$,
and abbreviate the sequence $(s_{j})_{j \neq i}$ to $s_{-i}$. We write
$(s'_i, s_{-i})$ to denote the joint strategy in which player $i$'s
strategy is $s'_i$ and for all $j \neq i$, player $j$'s strategy is
$s_j$.
Occasionally we write $(s_i, s_{-i})$ instead of
$s$.  Finally, we abbreviate the Cartesian product
$\times_{j \neq i} S_j$ to $S_{-i}$.
So in the degenerate situation when the game tree consists of just one
node, each strategy is the empty function, denoted by $\ES$, and there
is only one joint strategy, namely the $n$-tuple of these functions.
Each joint strategy assigns a unique descendant to every node in
$V \setminus Z$.  In fact, we can identify joint strategies with such
assignments.  

From now on the above notation will be used in the context of any
considered extensive game $G$. In particular $S_i$ will always denote
the set of strategies of player $i$.

Each joint strategy $s = (s_1, \LL, s_n)$ determines a rooted path
$\mathit{play}(s) := (v_1, \LL, v_m)$ in $T$ defined inductively as follows:

\begin{itemize}
\item $v_1$ is the root of $T$,

\item if $v_{k} \not\in Z$, then $v_{k+1} := s_i(v_k)$, where $turn(v_k) = i$.

\end{itemize}
So when the game tree consists of just one node, $v$, we have
$\mathit{play}(s) = v$.
Informally, given a joint strategy $s$, we can view $\mathit{play}(s)$ as the
resulting \emph{play} of the game.

Suppose now that the extensive game is well-founded.  Then for each
joint strategy $s$ the rooted path $\mathit{play}(s)$ is finite.  Denote by
$\leaf(s)$ the last element of $\mathit{play}(s)$. We call $(p_1(leaf(s)), \LL,
p_n(leaf(s)))$ the \bfe{outcome} of the game $G$ when each player $i$
pursues his strategy $s_i$ and abbreviate it as $p(leaf(s))$. We call
two joint strategies $s$ and $t$ \bfe{payoff equivalent} if
$p(leaf(s)) = p(leaf(t))$.

We say that a strategy $s_i$ of player $i$ a \bfe{best response} to a
joint strategy $s_{-i}$ of his opponents if for all $s'_i \in S_i$,
$p_i(leaf(s)) \geq p_i(leaf(s'_i, s_{-i}))$.
Next, we call a joint strategy $s$ a \bfe{Nash equilibrium} if each
$s_i$ is a best response to $s_{-i}$, that is, if
\[
\fa i \in \{1, \ldots, n\},  \fa s'_i \in S_i, \ p_i(leaf(s_i, s_{-i})) \geq p_i(leaf(s'_i, s_{-i})).
\]

\begin{example} \label{exa:ultimatum1}

  Let us return to the Ultimatum game from Example \ref{exa:ultimatum}.
  Each strategy for player 1 is a number,
  respectively from $[0,  100]$, while each
  strategy for player 2 assigns to every such number $x$ either $A$ or
  $R$.

  It is easy to check that each Nash equilibrium is of the form
  $(100, \ \mathrm{always} \ R)$, with the outcome  $(100, 0)$, 
 or
  $(x,s_2)$ with $s_2(x) = A$ and $s_2(y) = R$ for $y > x$, where
  $x,y \in [0,100]$, with the outcome  $(x, 100-x)$.
  \HB
\end{example}

Finally, we recall the notion of a subgame perfect equilibrium due to
Selten \cite{Sel65}(see also section 6.2 in \cite{OR94}), though now
defined for the larger class of well-founded games.

The \bfe{subgame of $G$ rooted at the node $w$}, denoted by
$G^w$, is defined as follows:

\begin{itemize}

\item its set of players is $\{1, \LL, n\}$,
  
\item its tree is $T^w$,

\item its turn and payoff functions are the restrictions of 
the corresponding functions of $G$ to the nodes of $T^w$.

\end{itemize}
Note that some players may `drop out' in $G^w$, in the sense that at
no node of $T^w$ it is their turn to move.  Still, to keep the
notation simple, it is convenient to admit in $G^w$ all original
players in $G$.

Each strategy $s_i$ of player $i$ in $G$ uniquely determines his
strategy $s^w_i$ in $G^w$.  Given a joint strategy
$s = (s_1, \LL, s_n)$ of $G$ we denote by $s^w$ the joint strategy
$(s^w_1, \LL, s^w_n)$ in $G^w$.  Further, we denote by $S_i^w$ the set
of strategies of player $i$ in the subgame $G^w$ and by $S^w$ the set
of joint strategies in this subgame.

Suppose now the extensive game $G$ is well-founded.  Then the notion
of a Nash equilibrium is well-defined.  A joint strategy $s$ of $G$ is
called a \bfe{subgame perfect equilibrium} in $G$ if for each node $w$
of $T$, the joint strategy $s^w$ of $G^w$ is a Nash equilibrium in
$G^w$.  Informally $s$ is subgame perfect equilibrium in $G$ if it
induces a Nash equilibrium in every subgame of $G$.

\begin{example} \label{exa:ultimatum2}

  Return now to the Ultimatum game from Example
  \ref{exa:ultimatum}. It is easy to check that it has exactly one
  subgame equilibrium, depicted in Figure \ref{fig:ultimatum1} by
  thick lines.

      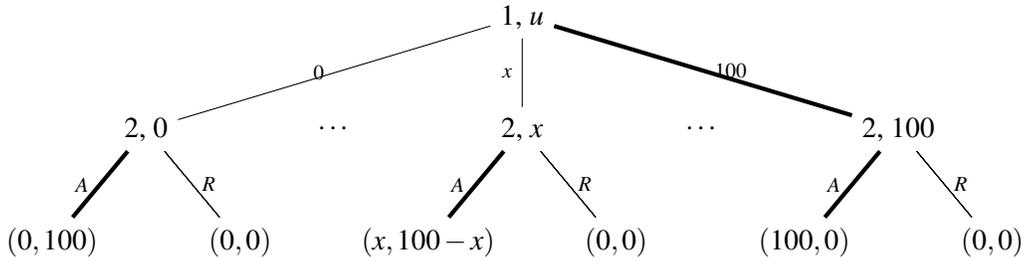
\begin{figure}[ht]
  \centering
  \tikzstyle{level 1}=[level distance=1.5cm, sibling distance=5cm]
  \tikzstyle{level 2}=[level distance=1.5cm, sibling distance=2.5cm]
  \tikzstyle{level 3}=[level distance=1.5cm, sibling distance=2cm]
\begin{tikzpicture}
 \node (r){1, $u$}
 child{
   node (a){2, $0$}
   child{
     node (d){$(0, 100)$}
     edge from parent [ultra thick]
     node[left]{\scriptsize $A$}
   }
   child{
     node(e){$(0,0)$}
     edge from parent
     node[right]{\scriptsize $R$}
     edge from parent
   }
   edge from parent
   node[left]{\scriptsize $0$}
   }
 child{
   node (b){2, $x$}
   child{
     node (f){$(x, 100-x)$}
     edge from parent [ultra thick]
     node[left]{\scriptsize $A$}
   }
   child{
     node (g){$(0,0)$}
     edge from parent [thin]
     node[right]{\scriptsize $R$}
     edge from parent
   }
   edge from parent 
   node[left]{\scriptsize $x$}
 }
  child{
   node (c){2, 100}
   child{
     node (h){$(100, 0)$}
     edge from parent [ultra thick]
     node[left]{\scriptsize $A$}
   }
   child{
     node (i){$(0,0)$}
     edge from parent
     node[right]{\scriptsize $R$}
     edge from parent [thin]
   }
  edge from parent [ultra thick]
  node[right]{\scriptsize $100$}};

  \path (a) -- (b) node [midway] {$\cdots$};
  \path (b) -- (c) node [midway] {$\cdots$};
\end{tikzpicture}
    \caption{The subgame perfect equilibrium in the Ultimatum game}
    \label{fig:ultimatum1}
      \end{figure}

Note that even though the rank of the game tree is 2, the customary 
backward induction cannot be applied here to compute 
this subgame equilibrium. Indeed, to deal with the root node the 
algorithm has to deal first with infinitely many nodes at which player 2 moves,
which leads to divergence.
\HB
\end{example}

An analysis of the subgame perfect equilibria in the games
from Examples \ref{exa:well-founded} and \ref{exa:alpha} is more
involved and will be provided in the next section using
a characterization of the set of subgame perfect equilibria of a
well-founded extensive game.

\section{Subgame perfect equilibria in well-founded games}
\label{sec:spe}
In general subgame perfect equilibria may not exist in well-founded
games. As an example take the modification of the Ultimatum game from
Example \ref{exa:ultimatum} in which instead of $[0, 100]$ one
considers the open interval $(0, 100)$.  In this section we establish
existence of subgame perfect equilibria in some natural classes of
well-founded games. This will directly follow from a characterization
of the sets of subgame perfect equilibria in such games.

We begin by stating a preparatory lemma, called the `one deviation
property' in \cite{OR94}. To keep the paper self-contained we include
in the Appendix the proof. It is more detailed than the one given in
\cite{OR94}.

\begin{restatable}{lemma}{lmSPE}
  \label{lm:SPE}
  Let $G$ be a well-founded extensive game over the game tree $T$. A
  joint strategy $s$ is a subgame perfect equilibrium in $G$ iff
  for all non-leaf nodes $u$ in $T$ and all $y \in C(u)$
 
  \begin{itemize}
  \item $p_i(\leaf(s^x)) \geq p_i(\leaf(s^{y}))$, where $i=\turn(u)$
    and $s_i(u)=x$.
  \end{itemize}
\end{restatable}

\begin{corollary} \label{cor:iff}
  \label{cor:spe}
Let $G$ be a well-founded extensive game over the game tree $T$ with the
root $v$. A joint strategy $s$ is a subgame perfect equilibrium in $G$
iff for all $u \in C(v)$

\begin{itemize}
  \item $p_i(\leaf(s^w)) \geq p_i(\leaf(s^u))$, where $i=\turn(v)$ and $s_i(v)=w$,

  \item $s^{u}$ is a subgame perfect equilibrium in the subgame $G^{u}$.
  \end{itemize}
\end{corollary}

Intuitively, the first condition states that among the subgames rooted
at the children of the root $v$, the one determined by the first move
in the game $G$ yields the maximal outcome for the player who moved.
Recall that for a function $f: X \to Y$ (with $X$ non-empty), $\mathrm{argmax}_{x \in X} f(x) := \{y \in X \mid f(y) = \max_{x \in X} f(x)\}$.
Using this notation this condition can be reformulated as: $s_i(v) \in \mathrm{argmax}_{u \in C(v)} p_i(\leaf(s^u))$, where $i=\turn(v)$.

\begin{proof}
If $C(v)=\emptyset$, the claim is vacuously true. Otherwise consider any
$u \in C(v)$.  By Lemma \ref{lm:SPE} $s^u$ is a subgame perfect
equilibrium in $G^u$ iff for all
non-leaf nodes $y$ in $T^u$ and $z \in C(y)$,
$
  p_i(\leaf((s^{u})^{x}) \geq p_i(\leaf((s^{u})^{z}),
$
where $i=\turn(y)$ and $s^u_i(y)=x$.

Since $(s^{u})^{x}=s^{x}$ and $(s^{u})^{z}=s^{z}$, the last statement is
equivalent to the statement that the inequality in Lemma \ref{lm:SPE}
holds for all non-leaf nodes $y$ in $T^u$ and $z \in C(y)$.
The conclusion now follows by Lemma \ref{lm:SPE}.
\end{proof}
\II

The above corollary allows us to characterize inductively the set of
subgame perfect equilibria in each well-founded extensive game.

Consider a well-founded extensive game $G$ with the root $v$ and
suppose $C(v) \neq \emptyset$.  Consider the subgames $G^{w}$, where
$w \in C(v)$, and a function $f$ that assigns to each sequence $t$ of
joint strategies in these subgames a child of $v$. Then each pair of
$t$ and $f$ determines a joint strategy in $G$ that we denote by
$(f,t)$.

Recall that by $S^w$ we denote the set of joint strategies in the subgame
$G^w$.  Given subsets $U^w$ of $S^w$ for $w \in C(v)$ and a set of
functions $F$ from $\times_{w \in C(v)} U^{w}$ to $C(v)$, we denote by
$[F, \times_{w \in C(v)} U^{w}]$ the set of joint strategies in $G$
defined by
\[
[F, \times_{w \in C(v)} U^{w}] := \begin{cases}
  \{(\ES, \LL, \ES)\} &\text{ if } C(v) = \emptyset \\
  \{(f, t)
  \mid f \in F, \: t \in \times_{w \in C(v)} U^{w} \} &\text{ otherwise }
\end{cases}
\]
In the first case $(\ES, \LL, \ES)$ stands for the joint strategy that
consists of the $n$-tuple of the empty strategies.
Note that when $C(v) \neq \emptyset$ if any of the sets $U^{w}$ or $F$
is empty, then so is $[F, \times_{w \in C(v)} U^{w}]$.
Further, we denote the set of subgame perfect equilibria in $G$ by
$\spe(G)$.

\vspace{-1mm}
\begin{theorem} \label{thm:rank}
  
  Consider a well-founded extensive game $G$ with the root $v$ and let $i=\turn(v)$.  Then
  \[
    \spe(G) = [F, \times_{w \in C(v)} \spe(G^{w})],
  \]
where if $C(v) \neq \ES$ then
$
  F = \{f \mid \forall t \in
  \times_{w \in C(v)} \: \spe(G^{w}) \ f(t) \in \mathrm{argmax}_{w \in
    C(v)} p_i(\leaf(t^w)) \}
$.
\end{theorem}

In particular, if the set
$\mathrm{argmax}_{w \in C(v)} p_i(\leaf(t^w))$ is empty, then
$F = \ES$ and hence $\spe(G) = \ES$.  Intuitively, each function
$f \in F$, given a sequence of subgame perfect equilibria in the
subgames rooted at the children of the root $v$, selects a root of
the subgame in which the outcome in the equilibrium is maximal for the
player who moves at $v$.  \II

\NI \textbf{Proof of Theorem \ref{thm:rank}.}  If $C(v) = \emptyset$,
then $(\ES, \LL, \ES)$ is a unique subgame perfect equilibrium, so the
claim holds.

If $C(v) \neq \emptyset$, then by Corollary
\ref{cor:iff} every subgame perfect equilibrium in the game $G$ is of
the form $(f,t)$, where for all $w \in C(v)$, $t^w$ is a subgame
perfect equilibrium in $G^{w}$ and for some $w \in C(v)$ we have
$f(t) = w$ and $p_i(\leaf(s^w)) \geq p_i(\leaf(s^u))$ for all
$u \in C(v)$.
\HB

\begin{corollary} \label{cor:finite}
  Every well-founded extensive game with finitely many outcomes has a subgame
  perfect equilibrium.
\end{corollary}

\begin{proof}
  The claim follows from Theorem \ref{thm:rank} by induction on the
  rank of the game tree and the observation that
  for every function $g: X \to Y$ with a finite range the set
$\mathrm{argmax}_{x \in X} g(x)$ is non-empty.
\end{proof}
\II

The above result can be generalized to some games with infinitely many
outcomes.  An example is a game in which for each player the set of
outcomes is either finite or equals the set of negative integers. More
generally, consider a well-founded extensive game in which for each
player the set of outcomes is a \emph{reverse well-ordered set}, i.e.,
every subset of this set has a greatest element.
Then Theorem \ref{thm:rank} implies that the game has
a subgame perfect equilibrium.

\begin{corollary} \label{cor:generic}
  Every well-founded extensive game without relevant ties has at most
  one subgame perfect equilibrium.
\end{corollary}
\begin{proof}
  If a game is without relevant ties, then so is every subgame of
  it. This allows us to proceed by induction on the rank of the game
  tree.  For game trees of rank 0 the claim clearly holds.  Suppose
  that it holds for all well-founded extensive games without relevant
  ties with the game trees of rank smaller than some ordinal
  $\alpha > 0$. Consider such a game with game tree of rank
  $\alpha$ and rooted at $v$. Let $i = \turn(v)$.

  By the induction hypothesis for each $w \in C(v)$ the set
  $\spe(G^{w})$ has at most one element.  If one of these sets is
  empty, then so is $\spe(G)$.

  So suppose that each $\spe(G^{w})$ is a singleton set. Then so is
  $\times_{w \in C(v)} \spe(G^{w})$. Let
  $\times_{w \in C(v)} \spe(G^{w}) = \{t\}$.  Then for different
  $w, w' \in C(v)$, $\leaf(t^w)$ and $\leaf(t^{w'})$ are different
  leaves of the game tree of $G$, so by the assumption about the game
  $p_i(\leaf(t^w)) \neq p_i(\leaf(t^{w'}))$, since $i = turn(v)$.

  This means that the function $g: C(v) \to \mathbb{R}$ defined by
  $g(w) := p_i(\leaf(t^w))$ is injective. Consequently the set
  $\mathrm{argmax}_{w \in C(v)} p_i(\leaf(t^w))$ has at most one
  element and hence the same successively holds for the sets $F$
  and $\spe(G)$.
\end{proof}
\II

In particular, every generic well-founded extensive game with finitely
many outcomes has a unique subgame perfect equilibrium.
We now show how Theorem \ref{thm:rank} can be used to reason about
subgame perfect equilibria in specific extensive games.

\begin{example}
  
  Consider the Bargaining game $G$ from Example
  \ref{exa:well-founded}.  Denote by $G(k)$ the game in which player 1
  first selects the number $k$.  The inductive structure of these
  games is depicted in Figure \ref{fig:well-founded}, where the
  actions of player 2 are $B$ (`make a better offer') or $A$ and $R$,
  as in Example \ref{exa:ultimatum}.

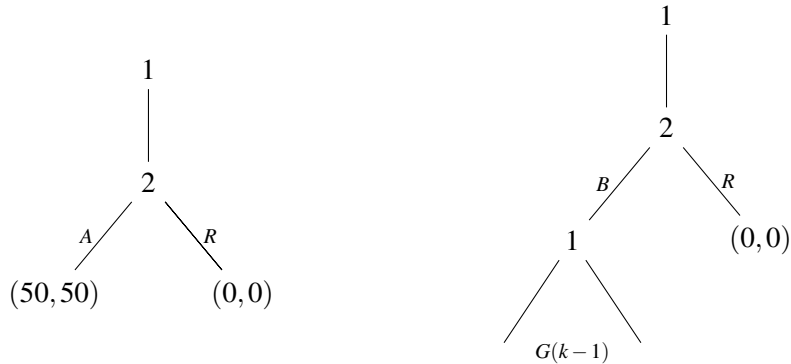
\begin{figure}[ht]
  \centering
  \tikzstyle{level 1}=[level distance=1.5cm, sibling distance=5cm]
  \tikzstyle{level 2}=[level distance=1.5cm, sibling distance=2.5cm]
  \tikzstyle{level 3}=[level distance=1.5cm, sibling distance=2cm]
  \begin{minipage}{.4\textwidth}
  \begin{tikzpicture}
    \node (r){1}
    child{
      node (b){2}
      child{
     node (f){$(50,50)$}
     edge from parent
     node[left]{\scriptsize $A$}
   }
   child{
     node (g){$(0,0)$}
     edge from parent
     node[right]{\scriptsize $R$}
     edge from parent
   }
   edge from parent
   node[left]{\scriptsize $$}
    };
  \end{tikzpicture}
  \end{minipage}
  \begin{minipage}{.4\textwidth}
   \begin{tikzpicture}
     \node (r){1}
     child{
       node (b){2}
       child{
         node (c){1}
         child{
           node(e){}
           edge from parent
         }
         child{
           node(f){}
           edge from parent
           }
         edge from parent
         node[left]{\scriptsize $B$}
       }
       child{
         node (d){$(0,0)$}
         edge from parent
         node[right]{\scriptsize $R$}
       }
       edge from parent
       node[left]{\scriptsize $$}
     };
     \path (e) -- (f) node [midway] {\scriptsize $G(k-1)$};
   \end{tikzpicture}
  \end{minipage}
  \caption{The games $G(2)$ and $G(k)$ for $k > 2$}
\label{fig:well-founded}
\end{figure}
   
It is easy to prove by induction using Theorem \ref{thm:rank} or
simply by the backward induction (the presentation of which we omit)
that each game $G(k)$, where $k \geq 2$, has a unique subgame perfect
equilibrium with the outcome $(50,50)$.

Children of the root of the game tree of $G$ are the roots of the game
trees of $G(k)$, where $k \geq 2$.  So for the game $G$ the set
$\mathrm{argmax}_{w \in C(v)} p_i(\leaf(t^w))$ referred to in Theorem
\ref{thm:rank} has exactly one element, 50.  Hence Theorem \ref{thm:rank}
implies that $G$ has a subgame perfect equilibrium, that the outcome
in each of them is (50,50), and that for each $k \geq 2$ there is 
a unique subgame perfect equilibrium in which player 1 first selects $k$.
\HB
\end{example}

\begin{example}

  Consider now the games $G(i, \alpha)$, where $i \in \{1, 2\}$ and
  $\alpha$ is an ordinal $> 1$ from Example \ref{exa:alpha}.  We
  noticed in Example \ref{exa:ultimatum2} that the game $G(1,2)$ has a
  unique subgame perfect equilibrium with the outcome (100,0). By
  symmetry the game $G(2,2)$ has a unique subgame perfect equilibrium
  with the outcome (0,100). The root of the game tree $G(i,3)$ has one
  child, which is the root of the game tree of $G(-i,2)$. Consequently
  $G(i,3)$ has a unique subgame perfect equilibrium with the outcome
  (0,100) for $i = 1$ and (100,0) for $i = 2$.

Using these observations we now show that for $i \in \{1, 2\}$ and 
ordinals $\alpha > 3$ the game $G(i, \alpha)$ has a
subgame perfect equilibrium and the outcomes in all these equilibria
are all (100,0) for $i = 1$ and (0,100) for $i = 2$.
We proceed by induction. Consider a game $G(1,\alpha)$ with
$\alpha > 3$ and assume the claim holds for all $\beta$ with
$3 \leq \beta < \alpha$.  The root of the game tree has as children
the roots of the game trees of $G(2, \beta)$, where
$1 < \beta < \alpha$.

By the induction hypothesis all these games except $G(2,3)$ have
subgame perfect equilibria with the outcomes (0,100).  For the game
$G(2,3)$, as just noted, the outcome in the unique subgame perfect
equilibrium is (100,0). So for the game $G(1, \alpha)$ the set
$\mathrm{argmax}_{w \in C(v)} p_i(\leaf(t^w))$ referred to in Theorem
\ref{thm:rank} has exactly one element, 100. Using this theorem we
conclude that the game $G(1,\alpha)$ has a subgame perfect equilibrium
and that the outcome in each equilibrium is (100,0).  A symmetric
claim, referring to (0,100) instead, holds for each game $G(2,\alpha)$
with $\alpha > 3$.

Using Theorem \ref{thm:rank} we conclude that each 
$G(i,\alpha)$ for $\alpha > 4$ has multiple subgame perfect equilibria.
\HB
\end{example}

Next, we establish a result showing that for a class of
well-founded extensive games all subgame perfect equilibria are payoff
equivalent.  The following condition was introduced in \cite{Roch80}:
\begin{equation}
  \label{equ:rochet}
\fa i \in \{1, \ldots, n\} \: \fa s, t \in S \ [p_{i}(\leaf(s)) = p_{i}(\leaf(t)) \to p(\leaf(s)) = p(\leaf(t))].
\end{equation}

This condition is in particular satisfied by the two-player well-founded extensive 
games that are \bfe{strictly competitive}, which means that 
\[
\fa i \in \{1,2\} \ \fa s,s' \in S \
p_{i}(\leaf(s)) \geq p_{i}(\leaf(s')) \text{ iff } p_{-i}(\leaf(s)) \leq p_{-i}(\leaf(s')).
\]
(To see it transpose $i$ and $-i$ and conjoin both equivalences.)

\begin{restatable}{theorem}{thmpayoff}
  \label{thm:payoff}
  In every well-founded extensive game
  that satisfies condition (\ref{equ:rochet})
all subgame perfect equilibria are payoff equivalent.
\end{restatable}
\begin{proof}
  First we prove the following claim.

\II
  
\NI
\textbf{Claim}. If a well-founded extensive game satisfies
condition (\ref{equ:rochet}), then so does every subgame of it.
\II

\NI
\emph{Proof}.
Let $G$ be a well-founded extensive game that satisfies condition
(\ref{equ:rochet}).  Consider any subgame $G^w$ of $G$.  Suppose that
for some player $i$ and joint strategies $s'$ and $t'$ in $G^w$ we
have $p_{i}(\leaf(s')) = p_{i}(\leaf(t'))$.  Take some joint strategies
$s$ and $t$ in $G$ such that
$\leaf(s) = \leaf(s'), \ \leaf(t) = \leaf(t'), \ s^w = s'$ and $t^w = t'$.
Then $p_{i}(\leaf(s)) = p_{i}(\leaf(t))$, so by condition
(\ref{equ:rochet}) $p(\leaf(s)) = p(\leaf(t))$ and consequently
$p(\leaf(s')) = p(\leaf(t'))$.
\HB
\II

We now proceed by induction on the rank of the game tree. For game
trees of rank 0 the claim obviously holds. Suppose
the claim holds for all well-founded extensive games whose game tree
is of rank smaller than some ordinal $\alpha >0$.
Consider a well-founded game $G = (T, turn, p_1, \LL, p_n)$
over a game tree of rank $\alpha$ with the root $v$.
Take two subgame perfect equilibria $s$ and $t$ in $G$.

If $path(s) = path(t)$, then $p(\leaf(s)) = p(\leaf(t))$. Otherwise take
the first non-leaf node $u$ lying on $path(s)$ such that
$s_i(u) \neq t_i(u)$, where $i = turn(u)$.  Let $s_i(u)= x$ and
$t_i(u) = y$.

Both $s^{y}$ and $t^{y}$ are subgame perfect equilibria in the
subgame $G^{y}$.  By the Claim the game $G^{y}$ satisfies condition
(\ref{equ:rochet}), so by the induction hypothesis $s^{y}$ and
$t^{y}$ are payoff equivalent in $G^{y}$.  We thus have
\[
  p_i(\leaf(s)) = p_i(\leaf(s^x)) \geq p_i(\leaf(s^{y})) = p_i(\leaf(t^{y})) =
  p_i(\leaf(t)),
\]
where the inequality holds by Lemma \ref{lm:SPE}.  Analogously
$p_i(\leaf(t)) \geq p_i(\leaf(s))$, so $p_i(\leaf(s)) = p_i(\leaf(t))$ and
hence by condition (\ref{equ:rochet}) $p(\leaf(s)) = p(\leaf(t))$.
\end{proof}
\II

For finite extensive games this result was stated in \cite[page
100]{OR94} as Exercise 100.2.  The most natural proof makes use of the
backward induction. For infinite games a different proof is needed.

We say that a well-founded extensive game $(T, turn, p_1, \LL, p_n)$
satisfies the \bfe{transference of decisionmaker indifference (TDI)}
condition if: $\fa i \in \{1, \ldots, n\} \: \fa r_i, t_i \in S_i \: \fa s_{-i} \in S_{-i}$,
\[p_{i}(\leaf(r_i, s_{-i})) = p_{i}(\leaf(t_i, s_{-i})) \to
p(\leaf(r_i, s_{-i})) = p(\leaf(t_i, s_{-i})).\]
Informally, this condition states that whenever for some player $i$ two of his strategies
$r_i$ and $t_i$ are indifferent w.r.t.~some joint strategy $s_{-i}$ of the other players
then this indifference extends to all players.

Clearly, condition (\ref{equ:rochet}) implies the TDI condition.  The TDI
condition was introduced in \cite{MS97}, the results of which imply
that in every finite extensive game with perfect information that
satisfies the TDI condition all subgame perfect equilibria are payoff
equivalent.  We conjecture that this result extends to well-founded
extensive games.

\section{Win or lose and chess-like games}
\label{sec:winorlose}

In this section we characterize subgame perfect equilibria of
two-player zero-sum well-founded extensive games with, respectively,
two and three outcomes. By Corollary \ref{cor:finite} each of these
games has a subgame perfect equilibrium. Below we consider the outcomes
$(1,-1)$, $(0, 0)$, and $(-1,1)$, but the obtained results hold with the same 
proofs for arbitrary outcomes as long as the game remains zero-sum.

A two-player extensive game is called a \bfe{win or lose
  game} if the only possible outcomes are $(1,-1)$ and $(-1,1)$, with
1 associated with winning and 0 with losing.  Given a well-founded win
or lose game $G$ we call a strategy $s_i$ of player $i$ a \bfe{winning
  strategy} if $\fa s_{-i} \in S_{-i} \ p_i(\leaf(s_i, s_{-i})) = 1$.
Below we denote the (possibly empty) set of winning strategies of
player $i$ in $G$ by $\win_i(G)$.

A classic result, attributed to Zermelo \cite{Zer13}, implies that in
finite win or lose games one of the players has a winning strategy.
This result also holds for arbitrary well-founded games.

\begin{theorem} \label{thm:win}
  Let $G$ be a well-founded win or lose game. For all players $i$ we
  have $\win_i(G) \neq \emptyset$ iff $\win_{-i}(G) = \emptyset$.
\end{theorem}
\begin{proof}
We have the following sequences of equivalences, where $i = turn(v)$:

\begin{tabular}{ll}
      & $s_i \in \win_i(G)$ \\
  iff &  \ \ \{ the definition of $\win_i(G)$ \} \\
      & for all $s_{-i} \in S_{-i}$, $p_i(\leaf(s_i,s_{-i}))=1$ \\ 
  iff &  \ \ \{ $i = \turn(v)$ \} \\
      & for all $s_{-i}^w \in S_{-i}^w$, $p_i(\leaf(s^w_i,s^w_{-i}))=1$,  where $w=s_i(v)$ \\
  iff &  \ \ \{ definition of a winning strategy \} \\
      & $s_i^w \in \win_i(G^w)$, where $w=s_i(v)$.
\end{tabular}
\II

\NI
and

\begin{tabular}{ll}
      & $s_{-i} \in \win_{-i}(G)$ \\
   iff &  \ \ \{ the definition of $\win_{-i}(G)$ \} \\
       & for all $s_{i} \in S_i$, $p_{-i}(\leaf(s_i,s_{-i}))=1$ \\
   iff &  \ \ \{ $i = \turn(v)$ \} \\
       & for all $w \in C(v)$ and $s_{i}^w \in S_{i}^w$, $p_{-i}(\leaf(s^w_i,s^w_{-i}))=1$ \\
   iff &  \ \ \{ definition of a winning strategy \} \\
       & for all $w \in C(v)$, $s_{-i}^w \in \win_{-i}(G^w)$.

\end{tabular}

  We now prove the claim by induction on the rank of the game tree.
  For game trees of rank 0 the claim clearly holds.
  Suppose that it holds for all well-founded win or lose games with
  game trees of rank smaller than some ordinal $\alpha > 0$ and consider a
  win or lose game $G$ with the well-founded game tree of rank
  $\alpha$ and rooted at $v$. Let $i = \turn(v)$.

  By the induction hypothesis for all $w \in C(v)$,
  $\win_i(G^w) \neq \ES$ iff $\win_{-i}(G^w) = \emptyset$, so the
  above equivalences imply the following string of equivalences:
\[
\mbox{$\win_i(G) \neq \emptyset$ iff for some $w \in C(v)$,
  $\win_i(G^w) \neq \ES$ iff for some $w \in C(v)$, $\win^w_{-i}(G) = \emptyset$ iff
    $\win_{-i}(G) = \emptyset$}
\]
and hence also $\win_{-i}(G) \neq \emptyset$ iff $\win_{i}(G) = \emptyset$.
\end{proof}
\II

From Corollary \ref{cor:finite} we know that every well-founded
win or lose game has a subgame perfect equilibrium, thus in particular
a Nash equilibrium. The following result clarifies the relation between
Nash equilibria and winning strategies.
We denote the set of Nash equilibria in an extensive game $G$ by
$\NE(G)$.

\begin{restatable}{corollary}{corNeWin}
  \label{cor:NeWin}
  Consider a well-founded win or lose game $G$.
  For some player $i$, $\NE(G) = \win_i(G) \times S_{-i}$.
\end{restatable}
\begin{proof}
By Theorem \ref{thm:win} $\win_1(G) \neq \ES$ or $\win_2(G) \neq \ES$.
Suppose without loss of generality that $\win_1(G) \neq \ES$.
\II

\NI
($\Ra$) 
Let $(s_1,s_2)$ be a Nash equilibrium and $t_1$ be a winning
strategy for player 1.
Then we have
$p_1(\leaf(s_1, s_2)) \geq p_1(\leaf(t_1, s_2)) = 1$ and hence
$p_2(\leaf(s_1, s_2)) = -1$.
If $s_1$ is not a winning strategy for player 1, then for some player 2
strategy $t_2$ we have $p_1(\leaf(s_1, t_2)) = -1$, i.e.,
$p_2(\leaf(s_1, t_2)) = 1 > p_2(\leaf(s_1, s_2))$, which contradicts the
fact that $(s_1,s_2)$ is a Nash equilibrium. So
$\NE(G) \sse \win_1(G) \times S_{2}$.  \II

\NI
($\La$) Take a winning strategy $s_1$ for player 1. Then for all strategies $t_1$ of player 1
and $s_2$ and $t_2$ of player 2,
\[
  p_1(\leaf(t_1, s_2)) \leq p_1(\leaf(s_1, s_2)) = p_1(\leaf(s_1, t_2)).
\]
So $(s_1,s_2)$ is a Nash equilibrium. Hence $\win_1(G) \times S_{2} \sse \NE(G)$.
\end{proof}
\II

In general the sets of subgame perfect equilibria and Nash equilibria
differ, so we cannot replace in the above result $\NE(G)$ by $\spe(G)$.
However, the above corollary directly implies the following characterization of
subgame perfect equilibria.

\begin{corollary}
Let $G$ be a well-founded win or lose game on a game tree $(V,E)$
  with the set of leaves $Z$. Then
    $\spe(G) = \{s \in S \mid \fa w \in V \setminus Z \: \te i \: [s^w
    \in \win_i(G^w) \times S^w_{-i}] \}$. 
\end{corollary}
It is easy to see that one cannot reverse here the order of the quantifiers.
\II

We now consider a related class of games often called \bfe{chess-like
  games}. These are two-player well-founded extensive games in which
the only possible outcomes are $(1,-1)$, $(0, 0)$, and $(-1,1)$, with
$0$ interpreted as a \emph{draw}.  We say that a strategy $s_i$ of
player $i$ in such a game \bfe{guarantees him at least a draw} if
\[
\fa s_{-i} \in S_{-i} \: p_i(leaf(s_i, s_{-i})) \geq 0,
\]
and denote the (possibly empty) set of such strategies 
by $\draw_{i}(G)$.

We now prove the following result for well-founded chess-like
games. The set $\win_{i}(G)$ is defined as above.

\begin{theorem} \label{thm:chess}
  In every well-founded chess-like game $G$
  \[
    \mbox{$\win_{1}(G) \neq \ES$ or $\win_{2}(G) \neq \ES$ or
      ($\draw_{1}(G) \neq \ES$ and $\draw_{2}(G) \neq \ES$).}
  \]
\end{theorem}

It states that in every chess-like game either one of the players has
a winning strategy or each player has a strategy that guarantees him
at least a draw.  These three alternatives are mutually exclusive,
since for all $i \in \{1,2\}$, $\win_{i}(G) \neq \ES$ implies both
$\win_{-i}(G) = \ES$ and $\draw_{-i}(G) = \ES$.

\begin{proof}
  We introduce the following abbreviations:

  \begin{itemize}
  \item $A$ for $\win_{1}(G) \neq \ES$,
    
  \item $B$ for $\draw_{2}(G) \neq \ES$,

  \item $C$ for $\win_{2}(G) \neq \ES$,

  \item $D$ for $\draw_{1}(G) \neq \ES$.
    
  \end{itemize}

Let $G_1$ and $G_2$ be the modifications of $G$ in which each
outcome $(0, 0)$ is replaced for $G_1$ by $(-1,1)$ and for $G_2$
by $(1,-1)$.  Then $\win_{1}(G_1) = \win_{1}(G)$, $\win_{2}(G_1) = \draw_{2}(G)$,
    $\win_{1}(G_2) = \draw_{1}(G)$, and $\win_{2}(G_2) = \win_{2}(G)$.

    Hence by Theorem \ref{thm:win} applied to the games $G_1$ and
    $G_2$ we have $A \lor B$ and $C \lor D$, so
    $(A \land C) \lor (A \land D) \lor (B \land C) \lor (B \land D)$,
    which implies $A \lor C \lor (B \land D)$, since
    $\neg (A \land C)$, $(A \land D) \equiv A$, and
    $(B \land C) \equiv C$.
\end{proof}

\II

For finite games, the above result is formulated in \cite[page
  125]{vNM04}. The proof first uses backward induction (apparently the
first use of it in the literature on game theory) to establish the
existence of a Nash equilibrium. Subsequently, (what is now called)
the Minimax theorem is invoked to conclude that the payoff to the
first player (and hence the second, as well) in any Nash equilibrium
is unique. Finally, it is observed that each possible payoff value
corresponds to one of the three disjuncts in the above theorem. The
above theorem clarifies that this result holds for well-founded
games as well, and that it can be proved in a simple way, without the
use of backward induction.

In \cite{Ewe02twoplayer} a proof of this result is provided for
chess-like games in which infinite plays, interpreted as draw, are
allowed. The proof does not rely on backward induction and is also
valid for well-founded chess-like games.

\begin{corollary}
  Consider a well-founded chess-like game $G$.
  For some player $i$
  \[
\mbox{$\NE(G) = \win_i(G) \times S_{-i}$ or $\NE(G) = \draw_i(G) \times \draw_{-i}(G)$.}
  \]
\end{corollary}

\begin{proof}
  Consider the games $G_1$ and $G_2$ from the proof of Theorem
  \ref{thm:chess}.  We noticed there that
  $\win_{1}(G_1) = \win_{1}(G)$, $\win_{2}(G_1) = \draw_{2}(G)$,
  $\win_{1}(G_2) = \draw_{1}(G)$, and $\win_{2}(G_2) = \win_{2}(G)$.

  So if $\win_1(G) \neq \ES$, then by Corollary \ref{cor:NeWin}
  applied to the game $G_1$ we get $\NE(G_1) = \win_1(G) \times S_2$,
  and if $\win_2(G) \neq \ES$, then by Corollary \ref{cor:NeWin}
  applied to the game $G_2$ we get $\NE(G_2) = S_1 \times \win_2(G)$.

  Suppose now that for both players $i$, $\win_i(G) = \ES$. Then both
  $\win_{1}(G_1) = \ES$ and $\win_{2}(G_2) = \ES$, so by Corollary
  \ref{cor:NeWin} applied to the games $G_2$ and $G_1$ we get both
  $\NE(G_2) = \draw_{1}(G) \times S_2$ and
  $\NE(G_1) = S_1 \times \draw_{2}(G)$. This implies
  $\NE(G_1) \cap \NE(G_2)  = \draw_{1}(G) \times \draw_{2}(G)$.
  
  Further, it is easy to see that $\NE(G) \sse \NE(G_1)$ and $\NE(G)
  \sse \NE(G_2)$.  Thus we have established that for some player $i$
  \[
\mbox{$\NE(G) \sse \win_i(G) \times S_{-i}$ or $\NE(G) \sse \draw_i(G) \times \draw_{-i}(G)$.}
\]

To complete the proof, let $p_i$ denote the payoff
function of player $i$ in the game $G$.
Suppose there exists a player $i$ such that $\win_i(G) \neq \emptyset$
and let $s \in \win_i(G) \times S_{-i}$.  By the definition of
$\win_i(G)$ for all $s'_{-i}$ we have
$p_i(leaf(s_i,s'_{-i}))=1$. Hence, since $G$ is a zero-sum game, for
all $s'_{-i}$ we have $p_{-i}(leaf(s_i,s'_{-i}))= -1$. Since 1 is a
maximum payoff for all $s'_i$ we also have
$p_i(leaf(s)) \geq p_i(leaf(s'_i,s_{-i}))$.  This shows that $s$ is a
Nash equilibrium of $G$.

Suppose now that for all players $i$, $\win_i(G) = \emptyset$. Fix
some $i \in \{1,2\}$ and let $s \in \draw_i(G) \times
\draw_{-i}(G)$. By the definition of the sets $\draw_i(G)$
\begin{itemize}
\item for all $s'_{-i}$, $p_i(leaf(s_i,s'_{-i})) \geq 0$ and
\item for all $s'_{i}$, $p_{-i}(leaf(s'_i,s_{-i})) \geq 0$.
\end{itemize}
Hence, since $G$ is a zero-sum game,  $p(leaf(s))=(0,0)$ and
\begin{itemize}
\item for all $s'_{-i}$, $p_{-i}(leaf(s_i,s'_{-i})) \leq 0$ and
\item for all $s'_{i}$, $p_{i}(leaf(s'_i,s_{-i})) \leq 0$.
\end{itemize}
This means that $s$ is a Nash equilibrium of $G$.
\end{proof}

\begin{corollary}
  Consider a well-founded chess-like game $G$ on a game tree $(V,E)$
  with the set of leaves $Z$. Then
  \[
    \spe(G) = \{s \in S \mid \fa w \in V \setminus Z \: \te i \: [s^w
    \in (\win_i(G^w) \times S^w_{-i}) \cup (\draw_i(G^w) \times
    \draw_{-i}(G^w))] \}.
  \]
\end{corollary}

\section{Conclusions}

In this paper we studied well-founded extensive games with perfect
information. We focused on the existence and structural
characterization of the sets of subgame perfect equilibria. We also
provided such characterizations for two classes of two-player zero-sum
games: win or lose games and chess-like games.  It will be interesting
to consider in this setting other notions and solution concepts that
have been well-studied in finite games.

One of them is weak dominance.  For finite games, its relation to
backward induction was studied in \cite{MS97}. The authors showed that
for finite game that satisfy the TDI condition from Section
\ref{sec:spe} the elimination of weakly dominated strategies is order
independent and is guaranteed to solve the game. The author of
\cite{Ewe02twoplayer,Ewe02} studied zero-sum extensive games and
showed that every such game with finitely many outcomes can be solved
by iterated elimination of weakly dominated strategies.

The definition of weak dominance applies to well-founded
extensive games, as well, but the resulting dynamics may be different. For
instance, it is possible that the iterated elimination of weakly
dominated strategies can then result in empty strategy sets for all or
for some players.  Also, it may happen that the elimination process
has to be iterated over ordinals larger than $\omega$.  It would be interesting
to identify subclasses of well-founded extensive games which can be
solved by the iterated elimination of weakly dominated strategies and
for which it is order independent.

Another direction is a study of the dynamics of strategy improvement
in terms of best (or better) response updates. For finite extensive
games, the relation between the improvement dynamics and Nash
equilibria was analyzed in \cite{Kuk02,BGHR17}. For restricted classes
of infinite games of perfect information, improvement dynamics were
studied in \cite{BEGM12,RP20}. It is an interesting question how the
improvement dynamics and Nash and subgame perfect equilibria relate in
well-founded extensive games.

\subsection*{Acknowledgements}
We would like to thank the reviewers and Marcin Dziubi\'{n}ski for
helpful comments. The second author was partially supported by the
grant MTR/2018/001244.

\bibliographystyle{eptcs}
\vspace{-5mm}
\bibliography{ref-s}
\section*{Appendix}

\lmSPE*
\begin{proof}
  \mbox{}
  
\NI ($\Ra$) Suppose $s$ is a subgame perfect equilibrium in $G$.
Consider a non-leaf node $u$ in $T$. Let $i=\turn(u)$, $x = s_i(u)$
and take some $y \in C(u)$. Let $t^u_i$ be the strategy obtained
from $s^u_i$ by assigning the node $y$ to $u$.

We now have
$p_i(\leaf(s^x)) = p_i(\leaf(s^u)) \geq p_i(\leaf(t^u_{i}, s^u_{-i}))
= p_i(\leaf(s^{y}))$, where the inequality holds by since $s^u$ a
Nash equilibrium in $G^u$.

\II

\NI
($\La$)
We proceed by induction on the rank of the game tree of $G$.
For game trees of rank 0 the induction hypothesis is vacuously true.
Suppose the claim holds for all well-founded extensive games
whose game tree is of rank smaller than some ordinal $\alpha >0$.
Consider a well-founded game $G$ over a game tree $T$ of rank
$\alpha$ with the root $v$.

Consider any node $u$ in $T$ such that $u \neq v$. (Since $\alpha >0$,
such a node $u$ exists.) Then $\mathit{rank}(T^u)$ is smaller than
$\alpha$ and for all nodes $w$ in $T^u$ we have $(s^u)^{w} = s^w$. By
the induction hypothesis $s^u$ is a subgame perfect equilibrium in
$G^u$, so a fortiori it is a Nash equilibrium in $G^u$. It remains to
prove that $s$ is a Nash equilibrium in $G$.

Suppose not. Then there exists player $i$ and $t_i \in S_i$ such that
for $t=(t_i,s_{-i})$ we have $p_i(\leaf(s)) < p_i(\leaf(t))$.
Recall that every joint strategy $s'$ in $G$ defines a rooted path
$\mathit{play}(s')$ in $T$. By the definition of $t$ these paths
differ for $s$ and $t$ at a node at which player $i$ moves. So for
some non-leaf node $u$ in $G$ with $\turn(u)=i$ we have
$\mathit{play}(s) = \sigma u x \pi_1$ and
$\mathit{play}(t) = \sigma u y \pi_2$, where $\sigma, \pi_1$ and $\pi_2$
are possibly empty sequences of nodes and $x \neq y$.
So $s_i(u)=x$ and  $t_i(u)=y$.

\NI
\emph{Case 1}. $s^y \neq t^y$.

Take the first, starting from the root, non-leaf node $w$ in $T^y$
such that $s_i(w) \neq t_i(w)$.  We have
$p_i(\leaf(s)) = p_i(\leaf(s^x)) \geq p_i(\leaf(s^y))= p_i(\leaf(s^w))
\geq p_i(\leaf(t^w)) = p_i(\leaf(t))$, where the first inequality
holds by the assumptions for the considered implication for the node
$v$ and the second by the fact that $s^{w}$ is a Nash equilibrium. So
we get a contradiction.
\II

\NI    
\emph{Case 2}. $s^y = t^y$.

We have
$p_i(\leaf(s^x)) = p_i(\leaf(s)) < p_i(\leaf(t)) = p_i(\leaf(t^y)) =
p_i(\leaf(s^y))$. But given that $s_i(u)=x$ this contradicts the
assumption for the node $u$.

This concludes the proof.
\end{proof}





\end{document}